\documentclass[journal,letterpaper,twocolumn]{IEEEtran}

\usepackage{amsfonts}
\usepackage{amssymb}
\usepackage{cite}
\usepackage[cmex10]{amsmath}
\usepackage{float}
\usepackage{color}
\usepackage{stfloats,fancyhdr}

\usepackage{algorithm}
\usepackage{algorithmic}
\usepackage{multirow}
\usepackage{changepage}
\usepackage[normalem]{ulem}

\newtheorem{proposition}{Proposition}

\IEEEoverridecommandlockouts

\ifCLASSINFOpdf
  \usepackage[pdftex]{graphicx}
  \DeclareGraphicsExtensions{.pdf,.jpeg,.png}
\else
  \usepackage[dvips]{graphicx}
  \DeclareGraphicsExtensions{.eps}
\fi

\usepackage{subfigure}

\usepackage{fancybox,dashbox}

\begin{document}

\title{{On the Performance of Multi-Antenna Wireless-Powered Communications with Energy Beamforming}
\thanks{Copyright (c) 2015 IEEE. Personal use of this material is permitted. However, permission to use this material for any other purposes must be obtained from the IEEE by sending a request to pubs-permissions@ieee.org.}
\thanks{This work was supported by the Australian Research Council (ARC) under
Grants DP150104019, FT120100487 and DP120100190. This work of H. Chen was supported by International Postgraduate Research Scholarship (IPRS), Australian Postgraduate Award (APA), and Norman I Price Supplementary scholarship.}
\thanks{The authors are with School of Electrical and Information Engineering, The University of Sydney, Sydney, NSW 2006, Australia (email: \{wenzhu.huang, he.chen, yonghui.li, branka.vucetic\}@sydney.edu.au). W. Huang and H. Chen have equal amount of technical contributions to this paper.}
\thanks{The corresponding author is H. Chen (email: he.chen@sydney.edu.au).}
}
\author{Wenzhu~Huang,~
He~(Henry)~Chen,~\IEEEmembership{Student~Member,~IEEE,}
Yonghui~Li,~\IEEEmembership{Senior~Member,~IEEE,} and Branka~Vucetic,~\IEEEmembership{Fellow,~IEEE}
}

\maketitle

\begin{abstract}
In this paper, we study the average throughput performance of energy beamforming in a multi-antenna wireless-powered communication network (WPCN). The considered network consists of one hybrid access-point (AP) with multiple antennas and a single antenna user. The user does not have constant power supply and thus needs to harvest energy from the signals broadcast by the AP in the downlink (DL), before sending its data back to the AP with the harvested energy in the uplink (UL). We derive closed-form expressions of the average throughput and their asymptotic expressions at high SNR for both delay-limited and delay-tolerant transmission modes. The optimal DL energy harvesting time, which maximizes the system throughput, is then obtained for high SNR. All analytical expressions are validated by numerical simulations. The impact of various parameters, such as the AP transmit power, the energy harvesting time, and the number of antennas, on the system throughput is also investigated.
\end{abstract}

\begin{IEEEkeywords}
RF Energy harvesting, energy beamforming, wireless-powered communication network, performance analysis, throughput.
\end{IEEEkeywords}

\IEEEpeerreviewmaketitle

\section{Introduction}

Limited energy supply has been a serious problem faced by many wireless communications systems due to a dramatic increase of their energy consumption. Recharging and replacing batteries may introduce a high cost and sometimes could be inconvenient or impossible, as for example, in dangerous and toxic environments or for sensors underwater or inside the human body. In this context, energy harvesting techniques, which can prolong the lifetime of energy constrained wireless networks in a sustainable way, have emerged and attracted significant attentions, e.g., see \cite{Ozel_JSAC_2011} and references therein. The techniques enable wireless nodes to harvest energy from the surrounding environment. Apart from other techniques that harvests energy from solar, wind, vibration, thermoelectric effects or other physical phenomena, radio frequency (RF) energy harvesting, which typically refers to the capability of the wireless nodes to scavenge energy from RF signals, has emerged as an attractive solution recently \cite{Visser_POI_2013,Lu_2014_Resource}. The feasibility of RF energy harvesting has been experimentally demonstrated by hardware implementation in \cite{Liu_2013_SIGCOMM} and \cite{powercast}. Wireless-powered communication networks (WPCNs), a new type of wireless networks emerging along with the RF energy harvesting technique, refer to networks in which the wireless devices normally have no internal energy supply (e.g., battery) and need to harvest energy from the RF signals emitted from a (dedicated) power transmitter before their information transmission \cite{Zhang_ICC_Tutorial}. In practice, WPCNs can find many potential applications, such as wireless sensor and RFID networks \cite{Smith_book_wirelss}.

There have been several research papers in open literature, which studied WPCNs for different network setups. Single-user WPCNs were investigated in \cite {Chen_TVT_2014_Wireless,Chen_WCL_2014_Energy}. Taking into consideration limited feedback, Chen \emph{et al.} derived the optimal time duration for the downlink (DL) wireless energy transfer (WET) through maximizing upper and lower bounds on the average information transmission rate of one-user WPCNs \cite {Chen_TVT_2014_Wireless}. Furthermore, the impact of channel estimation error on the average information transmission rate was also analyzed. \cite{Chen_WCL_2014_Energy} optimized the energy efficiency of a single-user WPCN by jointly designing the time duration and transmit power for DL wireless energy transfer. Multi-user WPCNs were first investigated in \cite{Ju_TWC_2014}, in which a ``harvest-then-transmit" protocol was developed. In the proposed protocol, the users first collect energy from the signals broadcast by a single-antenna hybrid access-point (AP) in the DL and then use the harvested energy to send independent information to the hybrid AP in the uplink (UL), based on the time-division-multiple-access (TDMA) scheme. The sum-throughput of the considered network, subject to user fairness, was maximized through optimizing the time allocation for DL WET and UL wireless information transfer (WIT). A similar network setup with a multi-antenna AP was investigated in \cite{Liu_arXiv_Multi}, where multiple users can simultaneously transmit information to the AP in the UL through space-division-multiple-access after they harvest energy in the DL. In \cite{Liu_arXiv_Multi}, the minimum throughput of all users was maximized by jointly optimizing the time allocation, the DL energy beams, the UL transmit power allocation, and the receive beamforming vectors. 

In contrast to the above papers, which mainly focused on resource allocation in WPCNs, in this paper we concentrate on analyzing the average performance of one-user multi-antenna WPCNs, in order to gain important insights for designing and implementing multi-antenna WPCNs in practice. Specifically, we derive the exact closed-form expression for the average throughput performance of a multi-antenna WPCN with energy beamforming.
It is also worth mentioning that there are several papers that analyzed the average performance of the simultaneous wireless information and power transfer (SWIPT) for different network scenarios \cite{Nasir_TWC_2013,Michalopoulos_arXiv_2013,Ding_TWC_2014,Morsi_arXiv_2014}, which are related but are essentially different to this paper. These papers focused on the characterization of the fundamental trade-offs between energy transfer and information transmission using the same signal. Moreover, \cite{Chen_TSP_2015,Ju_arXiv_2014,He_ITW_2014,Gu_ICC_2015} focused on the development of cooperative protocols for WPCNs with different setups. Note that all nodes were assumed to be equipped with single antenna in these papers. To the authors' best knowledge, this is the first paper that derives exact average performance of a multi-antenna WPCN with energy beamforming.

In this paper, two different transmission modes are considered for the user, i.e., delay-limited and delay-tolerant\footnote{The delay-tolerant mode is named as no-delay-limited in \cite{Liu_TWC_2013_Wireless}.}, which corresponds to different length of the code-words used by the user \cite{Liu_TWC_2013_Wireless,Nasir_TWC_2013}. More specifically, in the delay-limited mode, the received signal at the AP has to be decoded block by block and the code length should be no longer than the time duration of each transmission block. Thus, the average throughput of this transmission mode can be obtained via evaluating the outage probability. In contrast, in the delay-tolerant mode, the AP can store the received information blocks in a buffer and tolerate the delay for decoding the stored signals together. In this case, the code length can be very long compared to the transmission block time. The egordic capacity is then used to calculate the system average throughput in this delay-tolerant mode.

The main contributions of this paper are summarized as follows:
\begin{itemize}
  \item We derive exact closed-form expressions for the average throughput of the considered network in both delay-limited and delay-tolerant transmission modes by evaluating the outage probability and ergodic capacity with a given energy harvesting time.
  \item Asymptotic analyses are performed to gain insights and derive closed-form expressions of the throughput-optimal energy harvesting time at high SNR for both transmission modes. Futhermore, numerical simulations are presented to provide practical design insights into the impact of various parameters, such as the transmit power of the AP, the energy harvesting time, and the number of antennas, on the system performance.
\end{itemize}


\textbf{\emph{Notations:}} ${\left\| {\textbf{x}} \right\|}$ denotes the Euclidean norm of a vector ${\textbf{x}}$. $F_X\left(x\right)$ and $f_X\left(x\right)$ are used to represent the cumulative distribution function (CDF) and probability density function (PDF) of a random variable $X$, respectively.

The rest of the paper is organized as follows. Section II describes the system model. The exact and asymptotic performance analysis of the considered system is provided in Section III. Section IV presents the numerical results from which various design insights are given. Finally, Section V concludes this paper.


\section{System Model}
\begin{figure}
\centering
\scalebox{0.6}{\includegraphics{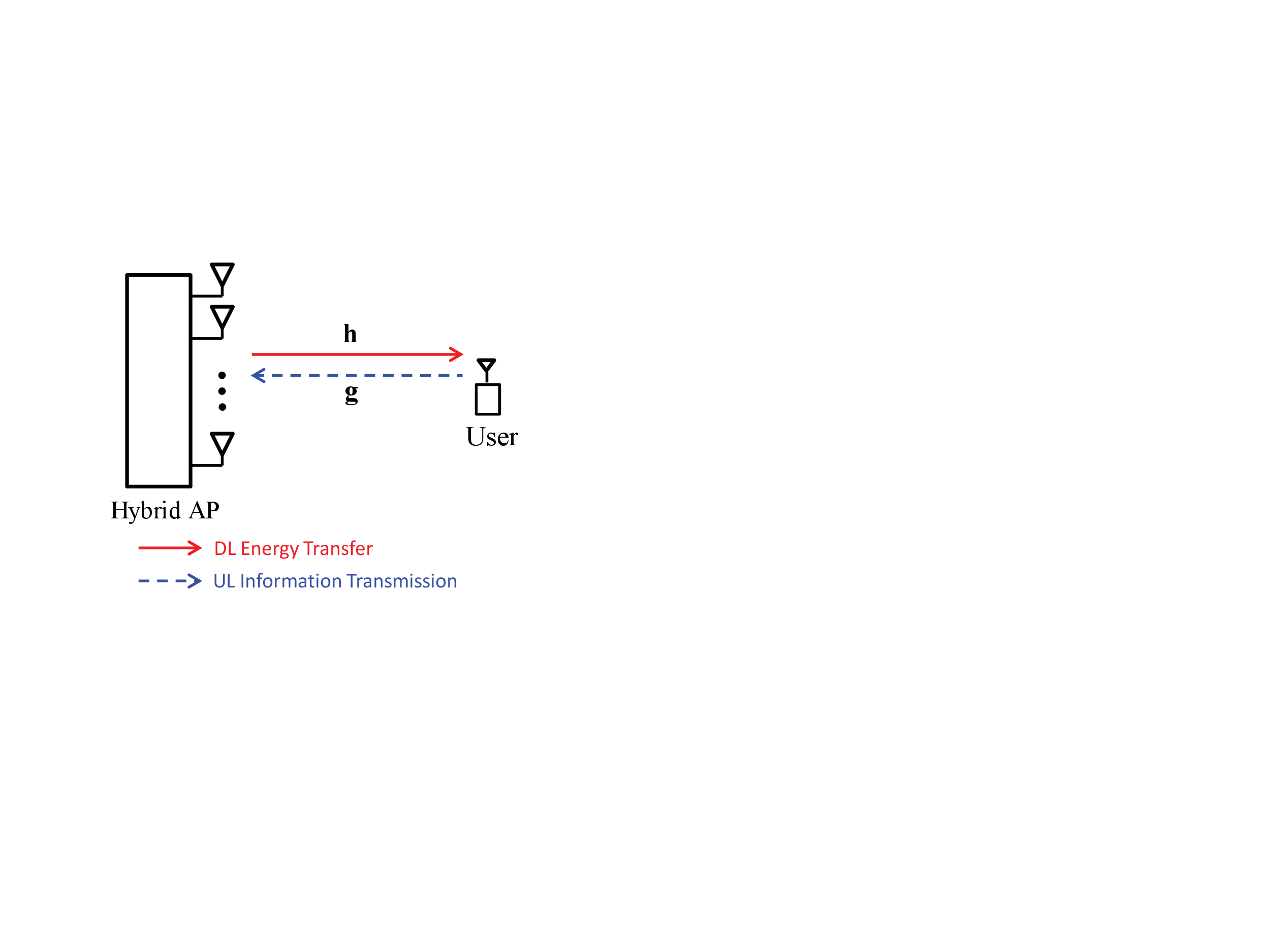}}
\caption{System model of a one-user multi-antenna WPCN.}
\label{fig:system_model}
\end{figure}

We study a single-user multi-antenna WPCN shown in Fig. \ref{fig:system_model}, with energy transfer in the downlink (DL) and information transmission in the uplink (UL). We follow \cite{Chen_TVT_2014_Wireless,Liu_arXiv_Multi} and consider that the AP is equipped with $N$ antennas but the user carries one single antenna\footnote{The extension to a general multiple-input multiple-output (MIMO) scenario with single/multiple user(s) \cite{Liu_arXiv_Multi,Zhang_conf_2014} would require the development of a new framework to analyze its system performance, which may constitute an independent paper and thus out of the scope of this paper.}. Such a network setup could correspond to a wireless sensor network in practice, where the user is a sensor node that cannot be equipped with multiple antennas due to the size and cost constraints. Both the AP and user work in a half-duplex mode. It is assumed that the user has no embedded energy supply and thus needs to first harvest energy from the signals broadcast by the AP in the DL and then use it for the information transmission to the AP in the UL. In our model, we consider a battery-free energy storage design at the user, which has been adopted in several recent works such as \cite{Chen_TVT_2014_Wireless,Ju_TWC_2014,Nasir_TWC_2013,Ding_TWC_2014}. Specifically, the user just uses a storage device such as a supercapacitor to keep the harvested energy among each transmission block.

The ``harvest-then-transmit" protocol proposed in \cite{Ju_TWC_2014} is assumed to implement in the considered network. Let $T$ denote the duration of one transmission block. The first interval of the duration $\tau T$ in each transmission block is used for the DL wireless energy transfer from the AP to the user and in the remaining interval of the duration $(1-\tau)T$, the user applies the harvested energy to transmit the information to the AP in the UL
. For simplicity but without loss of generality, we consider a normalized unit block time (i.e., $T = 1$) hereafter. In line with \cite{Chen_WCL_2014_Energy,Liu_arXiv_Multi}, we assume a quasi-static channel model with perfect channel state information (CSI) at the AP\footnote{In this paper, we consider the case that at the beginning of each transmission block, the AP first transfers a certain amount of energy to the user to acquire CSI. After this CSI acquisition duration, the AP delivers wireless energy to the user again, which will be used for UL information transmission. We assume that the channel is either static or varies very slowly, which is the case for wireless sensor networks, and thus the CSI acquisition duration is very short such that it can be ignored compared with the whole transmission block.}. In practice, the CSI can be acquired by various methods, e.g., the pilot-assisted reverse-link channel training \cite{Zeng_arXiv_2014_Zeng}. It is worth mentioning that it would be more general to consider the imperfect channel model as in \cite{Chen_TVT_2014_Wireless}. But this would require a totally different performance analysis framework, which may constitute an independent paper. Thus, we would like to consider this as our future work.

Let $\textbf{h}$ and $\textbf{g}$ represent the $N$ dimensional DL and UL channel vectors, respectively. Their entries are assumed to follow independent and identically distributed (i.i.d) circularly symmetric complex Gaussian distributions\footnote{Note that Rice fading is another practical model to describe the channel fading between the AP and user when the distance between them is limited and the light of sight exists. The performance analysis of the considered network over Rice fading has been treated as our future work.} with a zero mean and variance of $\Omega$. Let $P$ denote the transmit power from the AP during the DL phase. With the CSI $\textbf{h}$ available at the AP, the maximum ratio transmission (MRT) with beamforming vector $\textbf{w}^*= {{\textbf{h}}}/{{\left\| {\textbf{h}} \right\|}}$ is known to be optimal for DL energy transfer \cite{Chen_TVT_2014_Wireless}. By ignoring the negligible energy harvested from the noise, the amount of energy harvested by the terminal during the period $\tau $, denoted by $E_h$, is given by
\begin{equation}\label{eq:energy_harvested}
{E_h} = \eta \tau P {{\left\| {\textbf{h}} \right\|^2}},
\end{equation}
where $\eta$ is the energy conversion efficiency at the user.

After harvesting the energy from the AP, the user transmits the information to the AP during the remaining time $1 - \tau$. The received signal at the AP can be expressed as
\begin{equation} \label{eq2}
\textbf{y}_A = \sqrt {\frac{{{E_{h}}}}{{1 - \tau }}} \textbf{g}s + \textbf{n},
\end{equation}
where $s$ is the transmitted signal, and $\textbf{n}$ is the additive Gaussian white noise with a zero mean and variance matrix ${\sigma ^2}{{\bf I}_{{N}}}$. The maximum ratio combing (MRC) technique is assumed to be performed at the AP to maximize the received signal-to-noise ratio (SNR). The resulting SNR at the AP can be written as \cite{Chen_WCL_2014_Energy}
\begin{equation} \label{eq3}
\gamma_{A} = \frac{{{E_{h}} {{\left\| \textbf{g} \right\|}^2}}}{{\left( {1 - \tau } \right){\sigma ^2}}}
=\frac{\tau }{{1 - \tau }} {\bar \gamma} {{\left\| \textbf{h} \right\|}^2} {{\left\|\textbf{g} \right\|}^2},
\end{equation}
where $\bar \gamma  = \frac{{\eta  P }}{{{N_0}}}$.


\section{Performance Analysis}\label{sec:performance_analysis}

In this section, we analyze the average throughput performance of the considered system in both delay-limited and delay-tolerant transmission modes, respectively. This will be achieved by evaluating the outage probability and ergodic capacity of this network. We will also perform asymptotic analysis to obtain some insights on the effect of various system parameters. Moreover, we derive the expressions of the throughput-optimal energy harvesting time $\tau$ at high SNR for both transmission modes.

\subsection{Delay-Limited Transmission Mode}
\subsubsection{Exact Throughput Analysis}
In this transmission mode, the throughput can be obtained by evaluating the outage probability defined as the instantaneous channel capacity drops below the source's fixed transmission rate $R$ \cite{Nasir_TWC_2013}. Thus, the outage probability of the considered system with a given rate $R$ and energy harvesting time $\tau$ can be expressed as
\begin{equation}\label{eq:Pout}
\begin{split}
{P_{out}}(\tau) &= \Pr \left( {{{\log }_2}\left( {1 + {\gamma _A}} \right) < R} \right) 
= \Pr \left( {{\gamma _A} < {\gamma _0}} \right)\\
& = 1 - \frac{2\sum\limits_{p = 0}^{{N} - 1} {\frac{{
\left( \sqrt{ {\frac{(1- \tau) \gamma _0}{{\tau \bar \gamma \Omega^2}}} } \right)^{{N+p}}
{K_{{N} - p}}\left( {2\sqrt {\frac{(1- \tau) \gamma _0}{{\tau \bar \gamma \Omega^2}}} } \right)}}{{p!}}}}{{({N} - 1)!}},
\end{split}
\end{equation}
where $\gamma _{0} = 2^{R}-1$ and $K_n\left(  \cdot  \right)$ is the modified Bessel function of the second kind with order $n$ \cite[Eq. (9.6.2)]{Abramowitz_book_1970}. The last equation is obtained by noting that the terms  ${{\left\| \textbf{h} \right\|}^2}$ and ${{\left\| \textbf{g} \right\|}^2}$ in $\gamma_A$ follow independent and identical chi-square distributions with $2N$ degree of freedom \cite{Louie_ICC_2008_Perf}, and applying the PDF of product of two chi-square distributed variables derived in \cite{Shin_CL_2003_Effect}.

Considering that the user transmits with a fixed rate $R$ and the effective communication time from the user to AP after the energy harvesting phase is $1-\tau$, the average throughout of the delay-limited transmission mode, $\rho_{lim}$, is given by
\begin{equation} \label{eq:throughput_delay_limited}
\begin{split}
&\rho_{lim}  = \left( {1 - {P_{out}}(\tau)} \right)R\left( {1 - \tau } \right)\\
&= \frac{2R\left( {1 - \tau } \right)}{{({N} - 1)!}}\sum\limits_{p = 0}^{{N} - 1} {\frac{{
\left( \sqrt{ {\frac{(1- \tau) \gamma_0}{{\tau \bar \gamma \Omega^2}}} } \right)^{{N+p}}
{K_{{N} - p}}\left( {2\sqrt {\frac{(1- \tau) \gamma_0}{{\tau \bar \gamma \Omega^2}}} } \right)}}{{p!}}},
\end{split}
\end{equation}
where the throughput $\rho_{lim}$ depends on $\tau$, $N$, $\gamma_0$, $\Omega$, and $\bar \gamma$. Although the expression (\ref{eq:throughput_delay_limited}) gives exact value of the system throughput and may not be computationally complicated to evaluate, it does not offer explicit insights into the impacts of the system parameters. Thus, we are motivated to look into the system performance for this transmission mode at high SNR, where simpler expressions can be obtained.

\subsubsection{Asymptotic Throughput Analysis}
To this end, we can apply the series representation of Bessel functions in (\ref{eq:throughput_delay_limited}). But, due to complex structure of (\ref{eq:throughput_delay_limited}), this method is shown to be not tractable mathematically. To address this, let us first derive a tight approximation for (\ref{eq:throughput_delay_limited}). Recall that the end-to-end SNR $\gamma_A$ in (\ref{eq3}) can be regarded as a product of two chi-square random variables. In addition, the chi-square distribution can be regarded as a special case of the generalized gamma distribution. Recently, Chen \emph{et al.} proposed a novel analytical framework for evaluating the statistics of the product of independent random variables in \cite{Chen_TVT_2012}, where a tight approximation to the CDF of the product of generalized gamma random variables was provided. By applying the results obtained in \cite{Chen_TVT_2012}, we can approximate the CDF of $\gamma_A$ as
\begin{equation}\label{}
{F_{{\gamma _A}}}\left( z \right) \approx \frac{{\gamma \left( {{m_0} + 2N - 2,\frac{{2{m_0}}}{{{\Omega _0}}}\sqrt {\frac{z}{{{\Omega ^2}}}} } \right)}}{{\Gamma \left( {{m_0} + 2N - 2} \right)}},
\end{equation}
where $\gamma \left( {a,x} \right) = \int_0^x {{e^{ - t}}{t^{a - 1}}dt}$ is the lower incomplete gamma function and $\Gamma \left( a \right) = \int_0^\infty  {{e^{ - t}}{t^{a - 1}}dt}$ is the complete gamma function. Besides, $m_0 =  1.6467$ and $\Omega_0 = 1.5709$ are heuristically obtained \cite{Chen_TVT_2012}.

Consequently, the average throughput of the delay-limited transmission mode can be approximately expressed as
\begin{equation}\label{eq:throughput_delay_limited_approx}
\rho_{lim}  \approx R\left( {1 - \tau } \right)\left[ {1 - \frac{{\gamma \left( {{m_0} + 2N - 2,\frac{{2{m_0}}}{{{\Omega _0}}}\sqrt {\frac{{(1 - \tau ){\gamma _0}}}{{\tau \bar \gamma {\Omega ^2}}}} } \right)}}{{\Gamma \left( {{m_0} + 2N - 2} \right)}}} \right].
\end{equation}
As shown by numerical results in Sec. \ref{sec:num}, the above approximation is very tight.

In high SNR regime (i.e., $\bar \gamma \rightarrow \infty$), the term ${\frac{{2{m_0}}}{{{\Omega _0}}}}{\sqrt {\frac{{(1 - \tau ){\gamma _0}}}{{\tau \bar \gamma {\Omega ^2}}}} }$ inside the incomplete game function approaches to $0$. By applying the asymptotic property of the incomplete gamma function near zero $\gamma \left( {a,x} \right) \approx {x^a}/a$, we further have
\begin{equation}\label{eq:throughput_delay_limited_asymp}
\rho_{lim}  \approx R\left( {1 - \tau } \right)\left[ {1 - \frac{{{{\left( {\frac{{2{m_0}}}{{{\Omega _0}}}\sqrt {\frac{{(1 - \tau ){\gamma _0}}}{{\tau \bar \gamma {\Omega ^2}}}} } \right)}^{{m_0} + 2N - 2}}}}{{\Gamma \left( {{m_0} + 2N - 1} \right)}}
} \right].
\end{equation}

\emph{From (\ref{eq:throughput_delay_limited_asymp}), we now can observe the impact of the system parameters on the throughput performance. Specifically, when any of the parameters $N$, $\Omega$ and $\bar \gamma$ increases, the ratio term inside the square brackets of (\ref{eq:throughput_delay_limited_asymp}) will decrease, thereby improving the system throughput. However, the system throughput cannot keep increasing because the value of the term inside the square brackets is limited by $1$ no matter how large are the values of $N$, $\Omega$ and $\bar \gamma$. In other words, the system throughput in the delay-limited transmission mode should converge to a ceiling value of $R(1-\tau)$ by increasing the parameters $N$, $\Omega$ and $\bar \gamma$. Moreover, the impact of the parameter $\tau$ is not so explicit as that of the aforementioned ones. Particularly, the value of the term inside the square brackets rises when the value of $\tau$ increases. But, this will lead to the decrease of the term $(1-\tau)$ at the same time. Thus, we claim that there should exist an optimal energy harvesting time, denoted by $\tau_{lim}^*$, that can achieve the maximal throughput of the delay-limited mode.}

Based on the approximate system throughput (\ref{eq:throughput_delay_limited_approx}), we can derive an approximate expression for the value of $\tau_{lim}^*$ in high SNR regime, which is given in the following proposition:
\begin{proposition}\label{prop:optimal_tau_limited}
The throughput-optimal energy harvesting time at high SNR in the delay-limited transmission mode can be approximately expressed as
\begin{equation}\label{eq:optimal_tau_delay_limited}
\tau_{lim}^* \approx \frac{1}{{1 + {{\left\{ {\frac{{A + 2}}{B}W\left( { - \frac{B}{{A + 2}}{{\left( {\frac{{2\Gamma \left( A \right)}}{{{B^A}}}} \right)}^{1/\left( {A + 2} \right)}}} \right)} \right\}}^2}}},
\end{equation}
where $W(x)$ is the Lambert W function, which is the solution of the equation $We^W = x$. Besides,
$A = {m_0} + 2N - 2$ and $ B = \frac{{2{m_0}}}{{{\Omega _0}}}\sqrt {\frac{{{\gamma _0}}}{{\bar \gamma {\Omega ^2}}}} $ are defined for notation simplification.
\end{proposition}
\begin{proof}
See Appendix \ref{append:prop_optimal_tau_limited}.
\end{proof}

Note that the above expression for the optimal energy harvesting time is derived for high SNR regime. To obtain the optimal $\tau$ of any given SNR, we can calculate the first-order derivative of (\ref{eq:throughput_delay_limited}) with respect to $\tau$ and set it equal to zero. However, due to the complexity of (\ref{eq:throughput_delay_limited}), it is hard to achieve closed-form expressions for the roots of the obtained equation, which can only be calculated via numerical methods.

\subsection{Delay-Tolerant Transmission Mode}
\subsubsection{Exact Throughput Analysis}

In this subsection, we analyze the average throughput performance of the considered system with a delay-tolerant transmission mode. In this transmission mode, it is assumed that the AP can store the received information blocks in a buffer and tolerate the delay for decoding the stored signals together. Then, the average throughput can be determined by evaluating the ergodic capacity at the AP \cite{Nasir_TWC_2013}.

We first evaluate the ergodic capacity of the considered network for a given energy harvesting time $\tau$. According to the definition, the ergodic capacity of the system, denoted by $C\left(\tau\right)$, is given by
\begin{equation} \label{eq:capacity_def}
\begin{split}
C\left(\tau\right) &= \int_0^\infty  {{{\log }_2}\left( {1 + z} \right){f_{{\gamma _{_A}}}}\left( z \right)} dz\\
 &= \frac{2}{{(N - 1)!}\ln 2}\frac{{{{\left( {\frac{{1 - \tau }}{{\tau \bar \gamma {\Omega ^2}}}} \right)}^N}}}{{(N - 1)!}}\times \\
 &~~~\int_0^\infty  {{{\ln }}\left( {1 + z} \right){z^{N - 1}}{K_0}\left( {2\sqrt {\frac{{\left( {1 - \tau } \right)z}}{{\tau \bar \gamma {\Omega ^2}}}} } \right)} dz,
\end{split}
\end{equation}
where ${f_{\gamma_A}}\left( z \right)$ denotes the PDF of the random variable $\gamma_A$, which can be expressed as \cite{Shin_CL_2003_Effect}
\begin{equation} \label{eq:PDF_gamma_A}
\begin{split}
{f_{\gamma_A}}\left( z \right) = \frac{2{\left( {\frac{{1 - \tau }}{{\tau \bar \gamma \Omega^2}}} \right)^N}{z^{N - 1}}{K_0}\left( {2\sqrt {\frac{{\left( {1 - \tau } \right)z}}{{\tau \bar \gamma \Omega^2}}} } \right)}{{(N - 1)!}{(N - 1)!}}.
\end{split}
\end{equation}
Note that the ergodic capacity can also be calculated by using the CDF, which is, however, essentially the same with the adopted method using PDF.

Inspired by \cite{Matthaiou_conf_2011}, we can evaluate the integral term in (\ref{eq:capacity_def}) in a closed-form by applying the Meijer G-functions as follows. Firstly, we express the terms $\ln \left( {1 + x} \right)$ and ${x^a}{K_\nu }\left( x \right)$ in Meijer G-functions by using the equations given in \cite[Ch. 2]{Mathai_book_1973}.
Then, we can solve the integral and write (\ref{eq:capacity_def}) in a closed-form through the following calculation
\begin{equation}\label{eq:capacity_closed_form}
\begin{split}
&C\left(\tau\right)=\frac{{1 - \tau }}{{\tau \bar \gamma {\Omega ^2}}} \times\\
& \frac{{\int_0^\infty  {G_{2,2}^{1,2}\left( {z\left| {\begin{array}{*{20}{c}}
{1,1}\\
{1,0}
\end{array}} \right.} \right)G_{0,2}^{2,0}\left( {\left. {\frac{{\left( {1 - \tau } \right)z}}{{\tau \bar \gamma {\Omega ^2}}}} \right|N - 1,N - 1} \right)dx} }}{{(N - 1)!(N - 1)!\ln 2}}\\
&= \frac{{\frac{{1 - \tau }}{{\tau \bar \gamma {\Omega ^2}}}G_{2,4}^{4,1}\left( {\frac{{1 - \tau }}{{\tau \bar \gamma {\Omega ^2}}}\left| {\begin{array}{*{20}{c}}
{ - 1,0}\\
{ - 1, - 1,N - 1,N - 1}
\end{array}} \right.} \right)}}{{(N - 1)!(N - 1)!\ln 2}},
\end{split}
\end{equation}
where $G_{p,q}^{m,n}\left( {x\left| {\begin{array}{*{20}{c}}
   {{a_1}, \ldots ,{a_p}}  \\
   {{b_1}, \ldots ,{b_q}}  \\
\end{array}} \right.} \right)$ is the Meijer G-function \cite[Eq. (9.301)]{Gradshteyn_book_2000}, which is a
standard built-in function in most of the available mathematical software packages, such as MATLAB, MAPLE and MATHEMATICA. Besides, the integral of Meijer G-function is solved based on \cite[Eq. (7.811)]{Gradshteyn_book_2000}. Now, we can obtain the average throughput of the considered system, denoted by $\rho_{tol}$, by calculating the product of the ergodic capacity $C(\tau)$ and the effective time duration for information transmission of $1-\tau$. That is,
\begin{equation} \label{eq15}
\rho_{tol}  = \left( {1 - \tau } \right)C\left(\tau\right),
\end{equation}
where the average throughput $\rho_{tol}$ depends on $\tau$, $N$, $\Omega$ and $\bar \gamma $. However, due to the complexity of the Meijer G-functions, it is hard to observe the specific relationships between the throughput and the aforementioned parameters from (\ref{eq15}).

\subsubsection{Asymptotic Throughput Analysis}
To gain more insight, next let us derive the asymptotic throughput. For sufficiently high SNR, we can approximate ${\log _2}\left( {1 + z } \right)$ in the integral of (\ref{eq:capacity_def}) as ${\log _2}\left( z \right)$. The asymptotic throughput can be calculated without the need of Meijer G-function as
\begin{equation}\label{eq:asymptotic_throughput_DT}
\begin{split}
{\rho_{tol} } &\approx \frac{{2\left( {1 - \tau } \right)}}{{(N - 1)!}}\frac{{{{\left( {\frac{{1 - \tau }}{{\tau \bar \gamma {\Omega ^2}}}} \right)}^N}}}{{(N - 1)!\ln 2}}\times\\
&~~~\int_0^\infty  {{z^{N - 1}}\ln z\;{K_0}\left( {2\sqrt {\frac{{\left( {1 - \tau } \right)z}}{{\tau \bar \gamma {\Omega ^2}}}} } \right)} dz \\
&= \frac{{2\left( {1 - \tau } \right)}}{{(N - 1)!}}\frac{{{{\left( {\frac{{1 - \tau }}{{\tau \bar \gamma {\Omega ^2}}}} \right)}^N}}}{{(N - 1)!\ln 2}}\times \\
&~~~\int_0^\infty  {{t^{2N - 2}}2\ln t\;{K_0}\left( {2\sqrt {\frac{{\left( {1 - \tau } \right)}}{{\tau \bar \gamma {\Omega ^2}}}} t} \right)} 2tdt \\
&= \frac{{1 - \tau }}{{\ln 2}}\left[ {2\psi \left( N \right) + \ln {\bar \gamma} + 2\ln {\Omega } - \ln \left( {\frac{{1 - \tau }}{{\tau }}} \right)} \right],
\end{split}
\end{equation}
where \cite[Eq. (2.16.20.1)]{Prudnikov_book_1986} is used to solve the integral, $\psi \left(  \cdot  \right)$  is the Euler Psi function \cite[Eq. (8.36)]{Gradshteyn_book_2000}. \emph{From (\ref{eq:asymptotic_throughput_DT}), we can see that for a given value of $\tau$, the system throughput is proportional to the Psi function of the number of antennas $N$ at high SNR. This indicates that the system throughput increases with the rising of the number of antennas but the increasing rate is gradually decreasing, according to the properties of the Psi function. The average throughput is also proportional to the logarithm functions of the parameters $\bar \gamma$ and $\Omega$ in a high SNR regime. In addition, similar to the case in the delay-limited mode, we can clearly observe from (\ref{eq:asymptotic_throughput_DT}) that the energy harvesting time $\tau$ plays two opposite roles in the system throughput. Specifically, increasing the value of $\tau$ will decrease the value of the term outside the square brackets but increase the value of the term inside the square brackets. Thus, we also deduce that there exists an optimal value for the energy harvesting time, denoted by $\tau_{tol}^*$, that can maximize the system throughput.}

Based on (\ref{eq:asymptotic_throughput_DT}), we can obtain a closed-form expression for the optimal energy harvesting time $\tau_{tol}^*$ at high SNR for the delay-tolerant transmission mode, given in the following proposition:
\begin{proposition}\label{prop:optimal_tau_tolerant}
The throughput-optimal energy harvesting time at high SNR in the delay-tolerant transmission mode can be expressed as
\begin{equation}\label{eq:optimal_tau}
{\tau_{tol} ^*} {\approx} \frac{1}{{1 + W\left( {{\bar \gamma {\Omega ^2} e^{2\psi \left( N \right) - 1}}} \right)}}.
\end{equation}
\end{proposition}
\begin{proof}
See Appendix \ref{append:prop_optimal_tau_tolerant}.
\end{proof}
It is worth emphasizing that the above equation for the throughput-optimal energy harvesting time is actually a high-SNR approximate solution as it is obtained based on the high-SNR approximation in (\ref{eq:asymptotic_throughput_DT}).

\emph{From (\ref{eq:optimal_tau}), we can see that the value of the optimal energy harvesting time $\tau^*$ is inversely proportional to the parameters $\bar \gamma$, $\Omega$ and $N$, since the Lambert W function $W(x)$ is a monotonically increasing function for $x \ge 0$. This observation will be verified by simulation results in next section.}

\section{Simulation Results}\label{sec:num}
In this section, we provide simulation results to validate the derived analytical results. We also investigate the effects of the transmit power of the AP $P$, the number of antennas $N$, and the energy harvesting time $\tau$, on the system throughput. To capture the effect of path-loss on the network performance, we use the channel model that $\Omega = 10^{-3}\left({d_{AU}}\right)^{ - \alpha }$, where $d_{AU}$ is the distance between the AP and the user, and $\alpha \in [2,5]$ is the path-loss factor \cite{Chen_SPL_2010}. Note that a 30dB average signal power attenuation is assumed at a reference distance of $1$m in the above channel model \cite{Ju_TWC_2014}. In all following simulations, we set the distance between the AP and user $d_{AU} = 10$m, the path-loss exponent $\alpha$ = 2, the noise power $\sigma^2 = -80$dBm, the energy harvesting efficiency $\eta = 0.5$, and the transmission rate of the user $R = 2$ in the delay-limited mode.
\begin{figure*}
\centering
 \subfigure[Delay-limited mode]
  {\scalebox{0.5}{\includegraphics {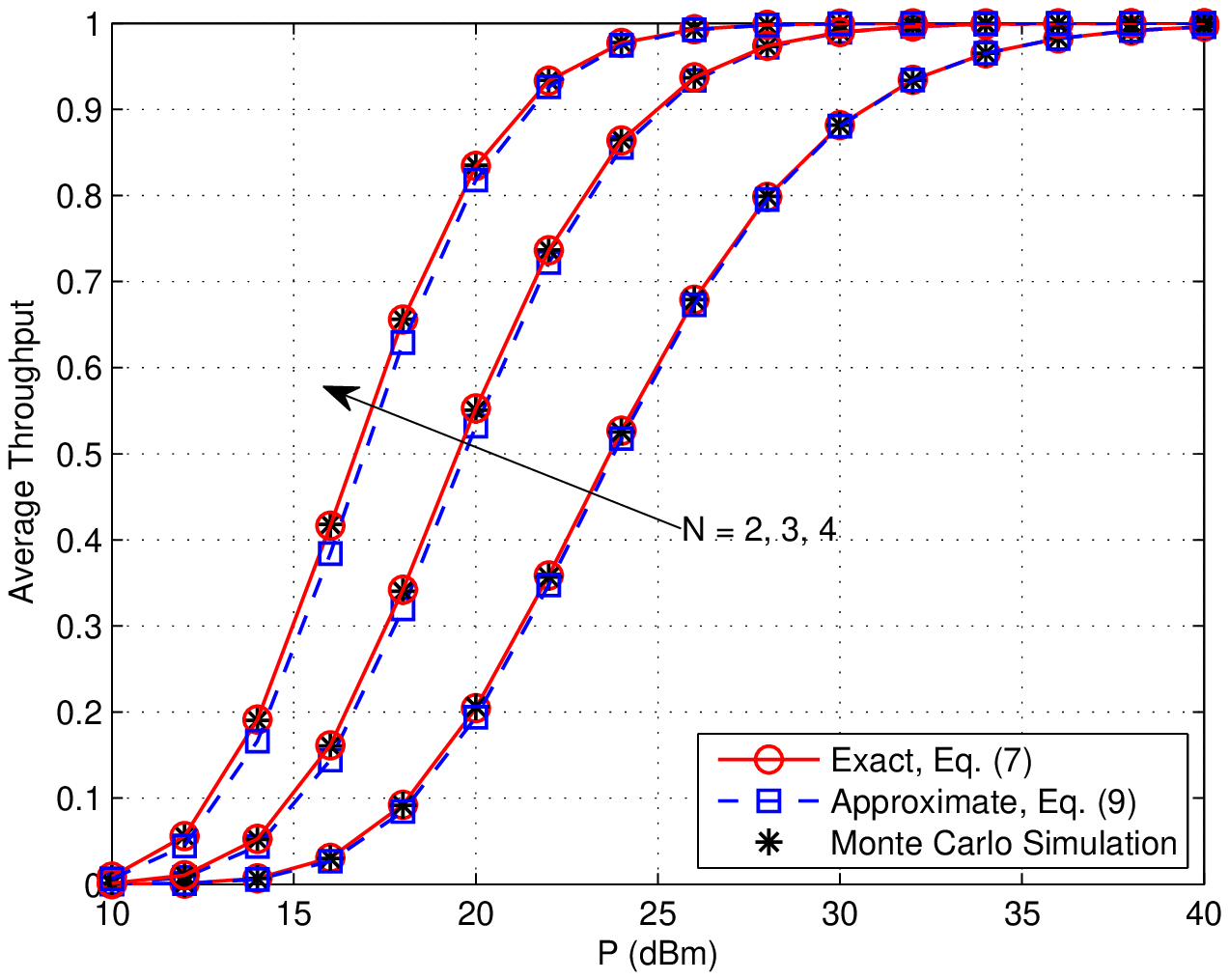}
  \label{fig:throughput_SNR_1}}}
\hfil
 \subfigure[Delay-tolerant mode]
  {\scalebox{0.5}{\includegraphics {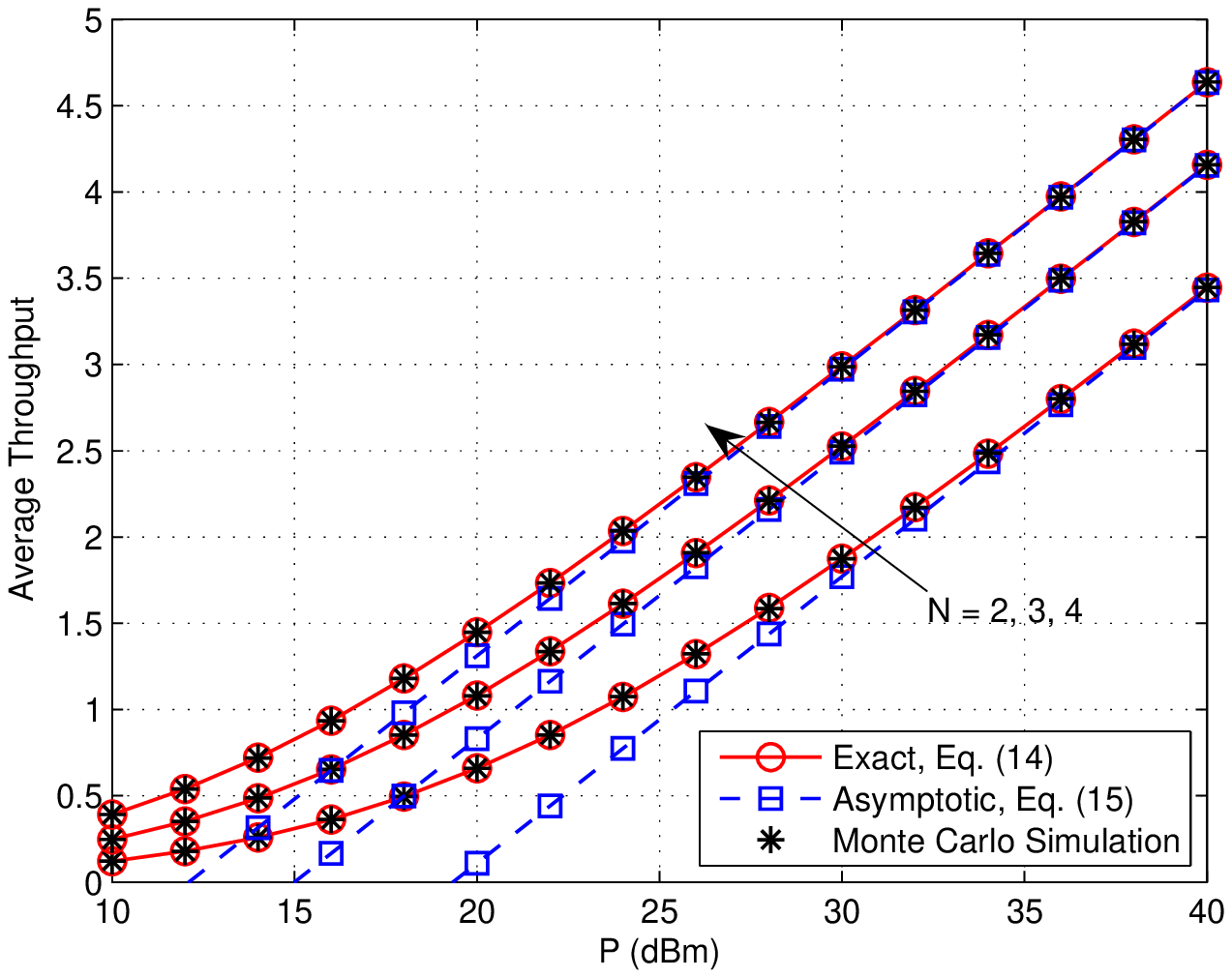}
  \label{fig:throughput_SNR_2}}}
\caption{Average throughput versus the transmit power of the AP for (a) delay-limited transmission mode and (b) delay-tolerant transmission mode with different number of antennas, where $\tau = 0.5$.}
\label{fig:throughput_SNR}
\end{figure*}

Fig. \ref{fig:throughput_SNR} plots the average throughput curves of the considered network versus the transmit power of the AP for both delay-limited and delay-tolerant modes with various numbers of antennas at the AP. We can see from Fig. \ref{fig:throughput_SNR} that the derived exact analytical results tightly coincide with the corresponding simulated ones, which verified the correctness of the derived closed-form expressions for the average throughput given in (\ref{eq:throughput_delay_limited}) and (\ref{eq15}). Besides, it can be observed from Fig. \ref{fig:throughput_SNR} (a) that the approximate throughput of the delay-limited mode calculated via (\ref{eq:throughput_delay_limited_approx}) can estimate the exact one very well, which confirms the effectiveness of the approximation. Moreover, for the delay-tolerant mode, the asymptotic throughput quickly approaches the exact one as the AP transmit power increases in both transmission modes, although there are gaps between them when the transmit power is low. This observation validates our asymptotic analysis in 
(\ref{eq:asymptotic_throughput_DT}). It can also be observed from this figure that the system throughput in both transmission modes is improved as either the transmit power or the number of antennas increases. However, the trends for the increasing of the system throughput are different in two transmission modes. Specifically, in the delay-limited transmission mode, the system throughput tends to be saturated when the transmit power of the AP is high enough. In addition, the larger the number of antennas at the AP, the lower the value of the transmit power from which the saturation of the system throughput starts. In contrast, the system throughput grows monotonically as the transmit power of the AP increases. These observations are actually consistent with our theoretical analyses discussed in Sec. \ref{sec:performance_analysis}.


\begin{figure}
\centering
\scalebox{0.55}{\includegraphics{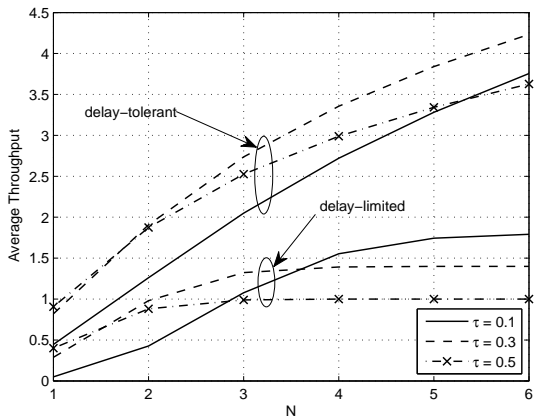}}
\caption{Average throughput versus the number of antennas at the AP for both delay-limited transmission mode and delay-tolerant transmission mode with different values of $\tau$, where $P = 30$ (dBm).}
\label{fig:throughput_N}
\end{figure}
Fig. \ref{fig:throughput_N} illustrates the effect of the number of antennas at the AP on the system throughput, in which the throughput curves of both transmission modes are plotted versus the number of antennas for different values of the energy harvesting time $\tau$. As we expected, the system throughput in delay-limited mode first increases and then converges to a ceiling as the number of antennas at the AP increases. This is because that the system throughput in this delay-limited mode is limited by the value of $R(1-\tau)$. On the other hand, the system throughput increases accordingly when the number of antennas becomes larger in the delay-tolerant mode. However, the slope of the throughput curves gradually decreases, which coincides with our theoretical analysis because the average throughput is proportional to the Psi function of the number of antennas. Moreover, we can see from Fig. \ref{fig:throughput_N} that the impact of the energy harvesting time $\tau$ on the system throughput of both transmission modes is not so clearly as that of the number of antennas does. For example, the throughput performance in the delay-tolerant mode is improved when the value of $\tau$ is increased to $0.3$ from $0.1$. However, as $\tau$ increases from $0.3$ to $0.5$, the system throughput for $N \ge 2$ significantly decreases.
\begin{figure*}
\centering
 \subfigure[Delay-limited mode]
  {\scalebox{0.5}{\includegraphics {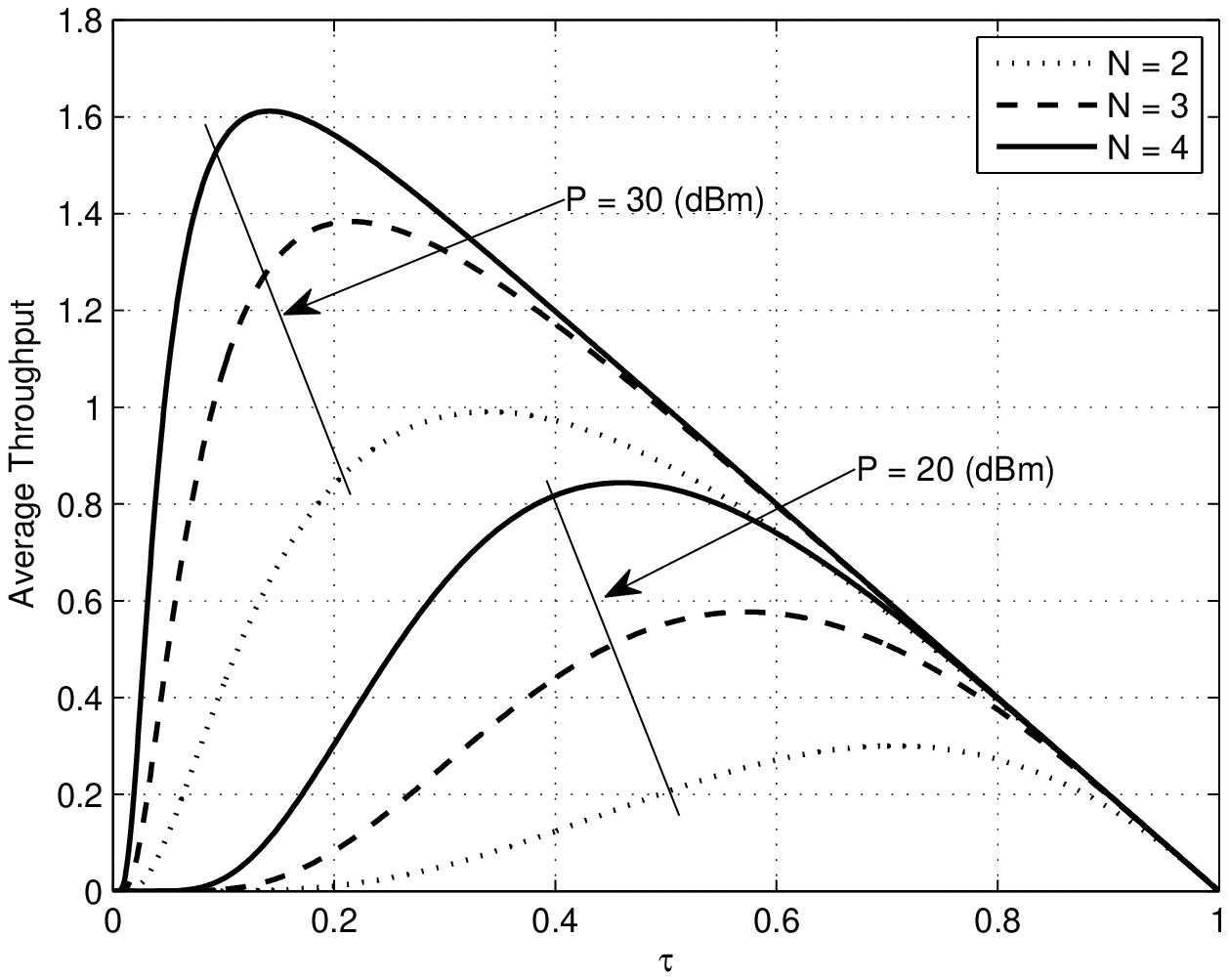}
  \label{fig:throughput_SNR_tau_1}}}
\hfil
 \subfigure[Delay-tolerant mode]
  {\scalebox{0.5}{\includegraphics {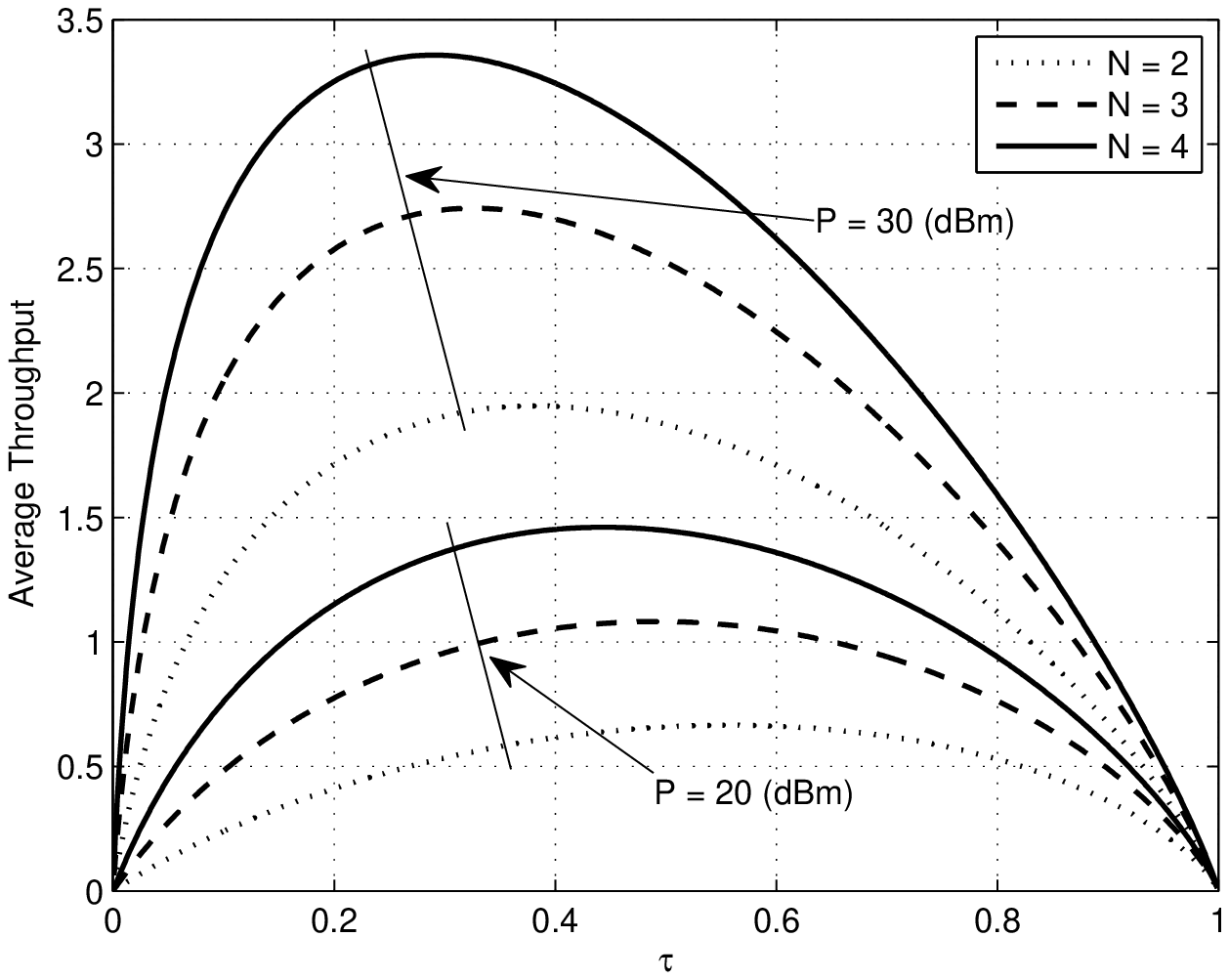}
  \label{fig:throughput_SNR_tau_2}}}
\caption{Average throughput versus the energy harvesting time $\tau$ for (a) delay-limited transmission mode and (b) delay-tolerant transmission mode with different values of $N$ and $P$.}
\label{fig:throughput_SNR_tau}
\end{figure*}

To clearly demonstrate the impact of the energy harvesting time $\tau$, we plot the throughput curves versus $\tau$ for both transmission modes with variable numbers of antennas and transmit powers of the AP in Fig. \ref{fig:throughput_SNR_tau}. It can be observed from this figure that there exists a throughput-optimal energy harvesting time for all the considered cases in both transmission modes, which coincides with our analytical results. Recall that we also derive the closed-form expressions for the optimal energy harvesting time at high SNR for both transmission modes in Sec. \ref{sec:performance_analysis}. To validate these two derived expressions given in (\ref{eq:optimal_tau_delay_limited}) and (\ref{eq:optimal_tau}), we compare them with the corresponding optimal energy harvesting time found via exhaustive search in Fig. \ref{fig:optimal_tau}. We can see from Fig. \ref{fig:optimal_tau} that the values of the optimal $\tau$, calculated by (\ref{eq:optimal_tau_delay_limited}) and (\ref{eq:optimal_tau}), coincide with the corresponding ones obtained via exhaustive search when the transmit power of the AP is high enough (i.e., at high SNRs), which verifies the correctness of our derivations of these two equations. In addition, as expected, the values of the optimal $\tau$ decrease as either the number of antennas or the transmit power increases. This is because the user can harvest the same amount of energy in a shorter time when either the number of antennas or the transmit power increases, and more time could be allocated to the UL information transmission to improve the system throughput.

\begin{figure*}
\centering
 \subfigure[Delay-limited mode]
  {\scalebox{0.5}{\includegraphics {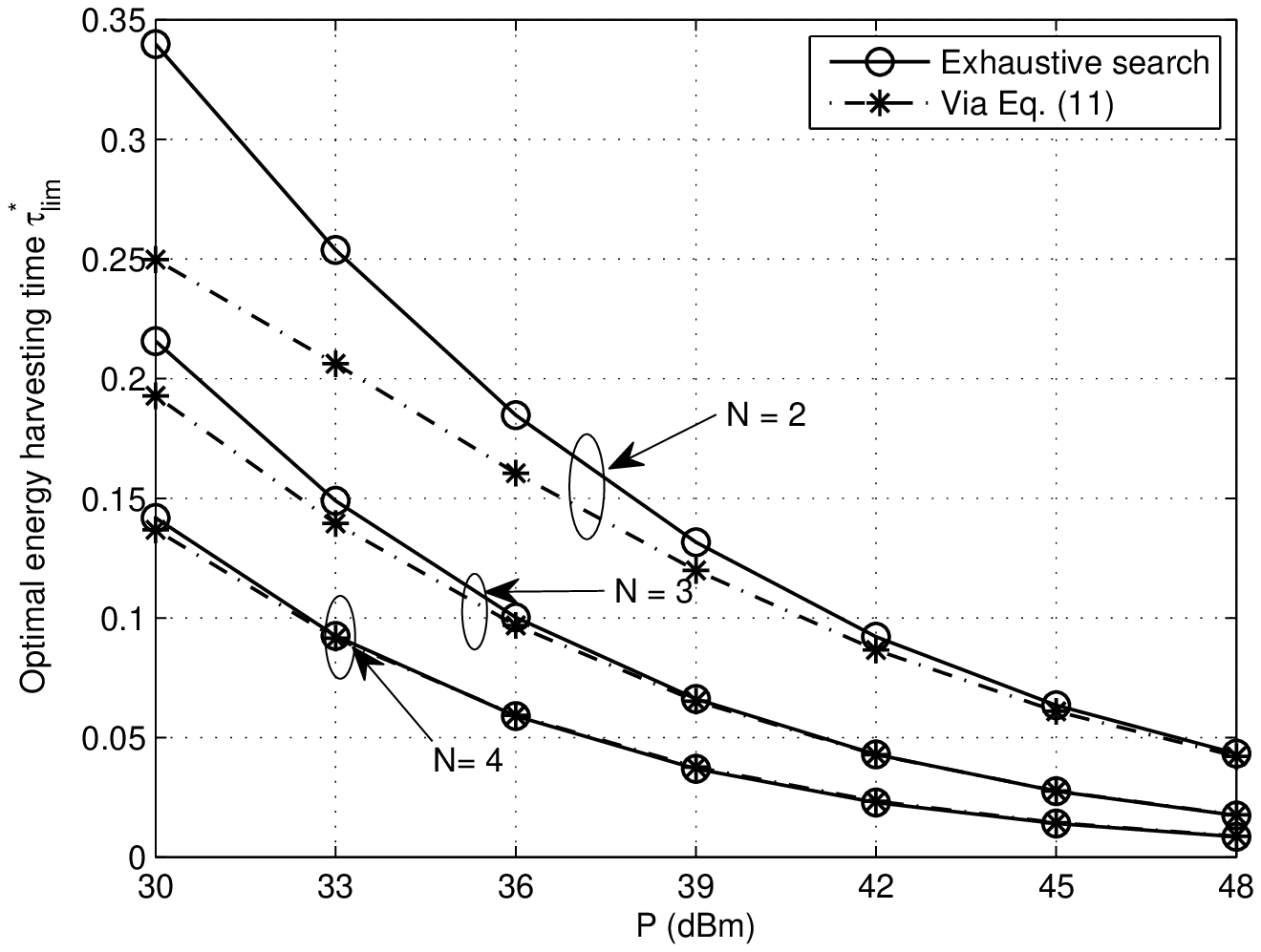}
  \label{fig:optimal_tau_1}}}
\hfil
 \subfigure[Delay-tolerant mode]
  {\scalebox{0.5}{\includegraphics {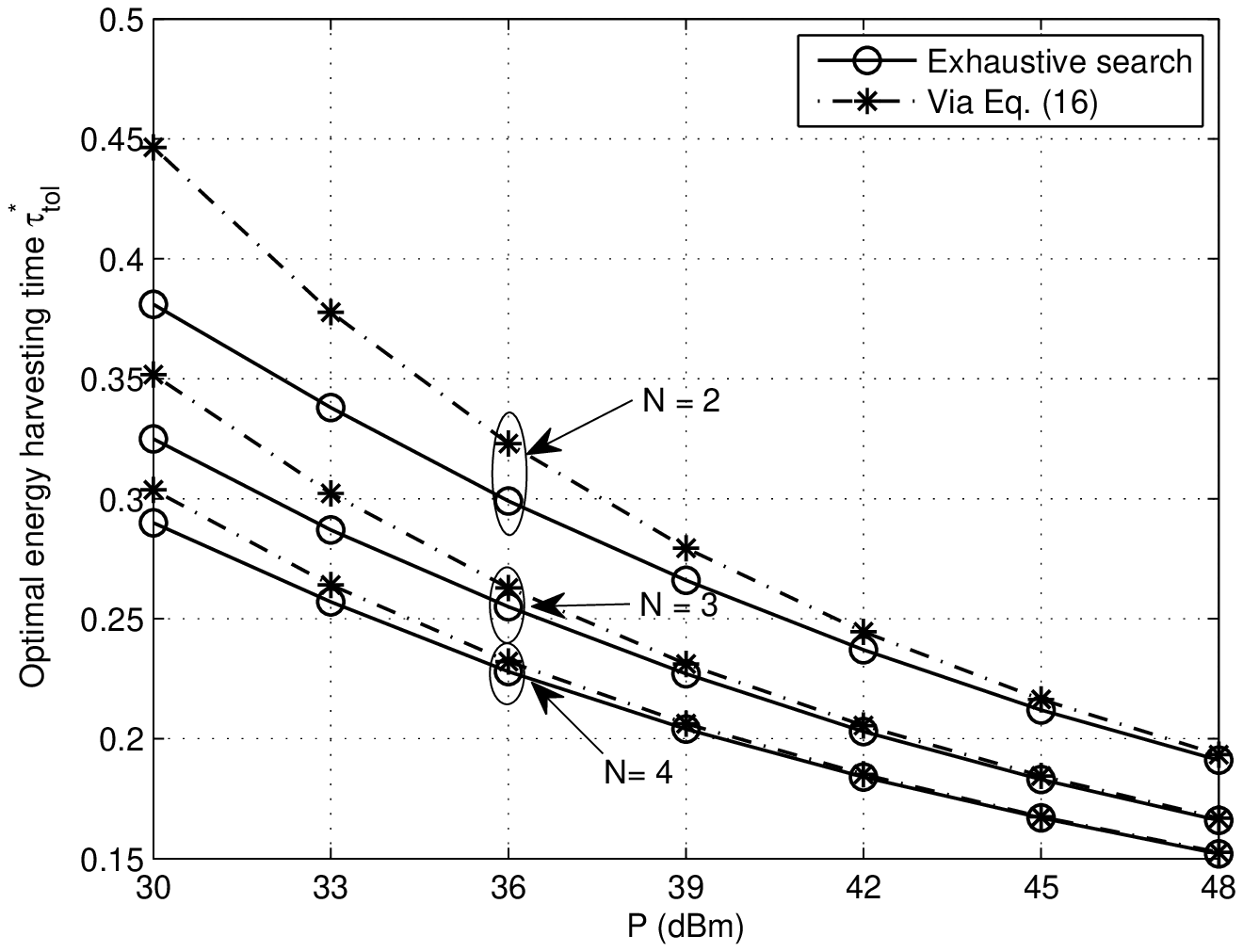}
  \label{fig:optimal_tau_2}}}
\caption{The optimal energy harvesting time versus the transmit power of the AP for (a) delay-limited transmission mode and (b) delay-tolerant transmission mode with different values of $N$.}
\label{fig:optimal_tau}
\end{figure*}

\section{Conclusion}
In this paper, we derived exact and asymptotic closed-form expressions for the average throughput of a multiple-antenna wireless-powered communication network with energy beamforming in both delay-limited and delay-tolerant transmission modes. The correctness and effectiveness of these theoretical analysis were verified by simulation results. Numerical results showed that the system throughput in the delay-limited mode first improves as either the number of antennas or the transmit power of the AP increases and then get saturated when these two parameters are large enough. In contrast, the system throughput in the delay-tolerant mode is always proportional to the number of antennas and the transmit power of the AP. Based on the asymptotic analysis, we also derived closed-form expressions of the throughput-optimal energy harvesting time at high signal-to-noise ratio (SNRs) for both transmission modes. Numerical simulations showed that the value of the optimal energy harvesting time is inversely proportional to the number of antennas and the transmit power of the AP. The optimal energy harvesting time calculated by the derived expressions can coincide well with the one obtained via exhaustive search when the SNR is high enough. Furthermore, the larger the number of antennas, the more accurate the optimal energy harvesting time calculated by the derived expressions.

\ifCLASSOPTIONcaptionsoff
  \newpage
\fi
\begin{appendix}

\subsection{Proof of Proposition \ref{prop:optimal_tau_limited}}\label{append:prop_optimal_tau_limited}
The approximate expression of the system throughput given in (\ref{eq:throughput_delay_limited_approx}) instead of the exact one is adopted to calculate the optimal energy harvesting time $\tau_{lim}^*$ since it is more tractable\footnote{Originally, the asymptotic expression (\ref{eq:throughput_delay_limited_asymp}) should be used to find the optimal energy harvesting time. However, due to the special polynomial structure of (\ref{eq:throughput_delay_limited_asymp}), we found it is hard to obtain a closed-form expression for the optimal energy harvesting time from (\ref{eq:throughput_delay_limited_asymp}). Thus, we choose to use (\ref{eq:throughput_delay_limited_approx}) instead.}. To proceed, we calculate the first-order derivative of (\ref{eq:throughput_delay_limited_approx}) with respect to $\tau$ and set it equal to $0$. After some algebraic manipulations, we have
\begin{equation}\label{eq:derivative_tau_delay_limited}
\gamma \left( {A,B\sqrt {\frac{{1 - \tau }}{\tau }} } \right) + \frac{1}{{2\tau }}{\left[ {B\sqrt {\frac{{1 - \tau }}{\tau }} } \right]^A}{e^{ - B\sqrt {\frac{{1 - \tau }}{\tau }} }} = \Gamma \left( A \right),
\end{equation}
where $A = {m_0} + 2N - 2$ and $ B = \frac{{2{m_0}}}{{{\Omega _0}}}\sqrt {\frac{{{\gamma _0}}}{{\bar \gamma {\Omega ^2}}}} $. For high SNR, the term $B$ approaches to 0. By applying the asymptotic property of the incomplete gamma function near zero $\gamma \left( {a,x} \right) \approx {x^a}/a$ once again and performing variable substitution $y = \sqrt{\frac{1-\tau}{\tau}}$, we can simplify (\ref{eq:derivative_tau_delay_limited}) as
\begin{equation}\label{eq:derivative_tau_delay_limited_1}
{y^A}\left( {1 + \frac{{A\left( {{y^2} + 1} \right)}}{2}{e^{ - By}}} \right) = \frac{{A\Gamma \left( A \right)}}{{{B^A}}}.
\end{equation}

Furthermore, in high SNR regime, only a small fraction of each transmission block is needed for DL energy transfer and the majority of time should be allocated to the UL information transmission. This means that $y = \sqrt{\frac{1-\tau}{\tau}}\gg 1$ should hold. In this case, we can ignore the two ``1" terms on the left-hand side of (\ref{eq:derivative_tau_delay_limited_1}) and obtain the following equation
\begin{equation}\label{eq:derivative_tau_delay_limited_2}
{y^{A + 2}}{e^{ - By}} = \frac{{2\Gamma \left( A \right)}}{{{B^A}}}.
\end{equation}
By equaling the $\{A+2\}$th roots of both sides in (\ref{eq:derivative_tau_delay_limited_2}), we have
\begin{equation}\label{}
\begin{split}
 &y{e^{ - \frac{B}{{A + 2}}y}} = {\left( {\frac{{2\Gamma \left( A \right)}}{{{B^A}}}} \right)^{\frac{1}{{A + 2}}}} \\
  \Leftrightarrow  &- \frac{B}{{A + 2}}y{e^{ - \frac{B}{{A + 2}}y}} =  - \frac{B}{{A + 2}}{\left( {\frac{{2\Gamma \left( A \right)}}{{{B^A}}}} \right)^{\frac{1}{{A + 2}}}} \\
  \Leftrightarrow  &- \frac{B}{{A + 2}}y = W\left( { - \frac{B}{{A + 2}}{{\left( {\frac{{2\Gamma \left( A \right)}}{{{B^A}}}} \right)}^{\frac{1}{{A + 2}}}}} \right).
\end{split}
\end{equation}
Substituting the relationship $y = \sqrt{\frac{1-\tau}{\tau}}$ into the last equation and rearranging it, we can obtain the desired results given in (\ref{eq:optimal_tau_delay_limited}).

\subsection{Proof of Proposition \ref{prop:optimal_tau_tolerant}}\label{append:prop_optimal_tau_tolerant}
To find the optimal energy harvesting time $\tau_{tol}^*$, we calculate the first-order derivative of the right-hand side of (\ref{eq:asymptotic_throughput_DT}) with respect to $\tau$ and set it to zero. After some algebraic manipulations, we obtain the following equation:
\begin{equation}\label{}
2\psi \left( N \right) + \ln \left( {\bar \gamma {\Omega ^2}} \right) - \frac{1}{\tau } = \ln \frac{{1 - \tau }}{\tau }.
\end{equation}
By performing an exponential operation on both sides of the above equation, we have
\begin{equation}\label{}
\begin{split}
 &\bar \gamma {\Omega ^2}{e^{2\psi \left( N \right)}}{e^{ - \frac{1}{\tau }}} = \frac{{1 - \tau }}{\tau } ,\\
  \Leftrightarrow\; &\bar \gamma {\Omega ^2}{e^{2\psi \left( N \right) - 1}}{e^{\frac{{\tau  - 1}}{\tau }}} = \frac{{1 - \tau }}{\tau }, \\
  \Leftrightarrow\; &\bar \gamma {\Omega ^2}{e^{2\psi \left( N \right) - 1}} = \frac{{1 - \tau }}{\tau }{e^{\frac{{1 - \tau }}{\tau }}}, \\
  \Leftrightarrow\;& \frac{{1 - \tau }}{\tau } = W\left( {\bar \gamma {\Omega ^2}{e^{2\psi \left( N \right) - 1}}} \right), \\
  \Leftrightarrow\;& \tau  = \frac{1}{{1 + W\left( {\bar \gamma {\Omega ^2}{e^{2\psi \left( N \right) - 1}}} \right)}}, \\
  \end{split}
\end{equation}
which completes the proof.
\end{appendix}

\end{document}